\documentclass[a4paper,english,cleveref, autoref, thm-restate]{lipics-v2021}

\usepackage [autostyle, english = american]{csquotes}
\MakeOuterQuote{"}

\newcommand{\Interact}{\operatorname{Interact}}
\newcommand{\mx}{\mathcal{M}^{\mathcal{X}}}
\newcommand{\my}{\mathcal{M}^{\mathcal{Y}}}
\newcommand{\A}{A}
\newcommand{\M}{\mathcal{M}}
\newcommand{\X}{\mathcal{X}}
\newcommand{\Y}{\mathcal{Y}}
\newcommand{\ed}{(\varepsilon,\delta)}

\newcommand{\remove}[1]{}

\pdfoutput=1 

\title{Improved Generalization Guarantees in Restricted Data Models} 

\author{Elbert Du\footnote{Corresponding Author}}{Department of Computer Science, Harvard University, USA}{elbertdu100@gmail.com}{}{}

\author{Cynthia Dwork\footnote{Corresponding Author}}{Department of Computer Science, Harvard University, USA}{dwork@seas.harvard.edu}{}{}

\authorrunning{E Du and C Dwork} 

\Copyright{Elbert Du and Cynthia Dwork} 



\category{} 

\relatedversion{} 

\acknowledgements{The authors are indebted to Guy Rothblum and Pragya Sur for many helpful conversations.}

\nolinenumbers 

\EventEditors{L. Elisa Celis}
\EventNoEds{1}
\EventLongTitle{3rd Symposium on Foundations of Responsible Computing (FORC 2022)}
\EventShortTitle{FORC 2022}
\EventAcronym{FORC}
\EventYear{2022}
\EventDate{June 6--8, 2022}
\EventLocation{Cambridge, MA, USA}
\EventLogo{}
\SeriesVolume{218}
\ArticleNo{7}

\begin{document}
\maketitle

\begin{abstract}
Differential privacy is known to protect against threats to validity incurred due to adaptive, or exploratory, data analysis -- even when the analyst adversarially searches for a statistical estimate that diverges from the true value of the quantity of interest on the underlying population. The cost of this protection is the accuracy loss incurred by differential privacy.  
In this work, inspired by standard models in the genomics literature,  we consider data models in which individuals are represented by a sequence of attributes with the property that where distant attributes are only weakly correlated.   We show that, under this assumption, it is possible to "re-use" privacy budget on different portions of the data, significantly improving accuracy without increasing the risk of overfitting.

\end{abstract}
\section{Introduction}

It has been known for nearly a decade that interacting with data in a differentially private fashion provides a universal approach to reducing the risk of spurious scientific discoveries incurred by {\em adaptive}, or {\em exploratory}, data analysis~\cite{dwork2014preserving,dwork2015preserving}, in which new analyses or questions posed of the data depend on the outcomes of previous analyses.  
Strengthenings of these initial results, and extensions to other information-restrictive interactions, rapidly followed, for example, \cite{bassily2021algorithmic,dwork2015generalization}.  In these works and their {\it sequelae}, the data analyst is viewed as an {\em accuracy adversary} whose goal is to find a query on which the dataset (or the response produced by a mechanism that interacts with the data) is not representative of the population.

For some kinds of data and analyses, for example, in Genome-Wide Association Studies (GWAS), which involve vast numbers of statistical queries on very high dimensional data, differential privacy faces daunting lower bounds~\cite{bun2018fingerprinting}.  However, our interest in this work is in accuracy, and not privacy {\it per se}. 
Inspired by two natural examples, we consider the question of whether we can improve on the accuracy by exploiting independence properties in the features of the data.  In data streams, it is often assumed that elements far apart in the stream are uncorrelated or only weakly correlated, with the correlation decreasing as the distance increases. In a stream, data of different individuals are interleaved; genomic information has this same low-correlation property even in the DNA for a single individual: for example, chromosomes are considered to be unrelated, and even within a chromosome correlations decrease with distance \cite{Slatkin2008LinkageD}.

While genomic data is our motivating example, we note that similar assumptions are reasonable in other settings.  For example, in certain kinds of image data distant pixels may be relatively uncorrelated even within a single image.  We will make this notion precise in Section~\ref{sec:preliminaries}.

The line of work described above gave rise to a number of so-called "transfer theorems," and we will make use of the sharp recent addition to this literature in~\cite{jung2019new}.  Transfer theorems generally say that if a query-response mechanism satisfies some specific quantifiable constraint on the information it imparts, then an analyst interacting with this mechanism cannot overfit to within some related quantity.  In the context of differential privacy the requirement is that the mechanism must be $(\varepsilon,\delta)$-differentially private and $(\alpha',\beta')$-sample accurate\footnote{That is, with probability at least $1-\beta'$ the responses produced are within $\alpha'$ of their sample values.}, and the guarantee from the theorem is that the responses will be $(\alpha=\alpha(\epsilon,\delta,\alpha'), \beta=\beta(\epsilon,\delta,\beta'))$-distributionally accurate, meaning that with probability at least $1-\beta$ the responses are within $\alpha$ of their distributional values. 

Our restriction on data models comes into play here: consider a genome-wide association study (GWAS), in which the dataset contains, for each of $n$ individuals, a string of potentially millions of Single Nucleotide Polynomorphisms (SNPs).  A typical study will make huge numbers of counting queries, looking for SNPs that are associated with a disease, at a huge cost in accuracy, as the data of each individual simultaneously affect all these counts. We asked the following question: under the assumption that distant SNPs in the genome of any given individual are at best very loosely correlated, is it possible to "re-use" privacy budget when examining distant portions of the genome? We will not achieve privacy in so doing, but can we achieve better accuracy?  For example, if we examine the dataset one chromosome at a time, meaning, we analyze the first chromosome for everyone in the dataset using $(\varepsilon_0,\delta_0)$-DP and a single application of a transfer theorem to ensure validity on the queries for this chromosome, and then examine the second chromosome for everyone in the dataset, "re-using" $(\varepsilon_0,\delta_0)$-DP, and it really is the case that one's first and second chromosomes are unrelated, can we safely apply the transfer theorem a second time to conclude that the queries on the second chromosome have not overfit, and so on?  We obtain an affirmative answer to this and other, less restrictive, data access models.  The key factors in the analysis are (1) the independence of the features (chromosomes, distant SNPs) and (2) the exclusion of queries that simultaneously operate on distant features (sums of adjacent features permitted, sums of distant features not supported). 

Our first result considers the model in which each individual's data is partitioned into a sequence of $m$ fully independent blocks.  Roughly speaking, it says that the privacy budget for a single block can be re-used, risking only a factor of $m$ increase in failure probability.

\begin{theorem}
(Informal) If the data consists of $m$ independent blocks, and our mechanism $M$ performs an $(\epsilon, \delta)$-DP and $(\alpha, \beta)$-sample accurate interaction on each block, then $M$ is $(\alpha', m\beta')$-distributionally accurate, where $\alpha'$ and $\beta'$ are the parameters we get from the transfer theorem on each block.
\end{theorem}
To build intuition for this result, suppose that, for each individual, we have a series of $m>1$ mutually independent blocks of features $B_1,B_2,...B_m$.  That is, there are $m$ distributions $D_1,\dots,D_m$ and the data of each individual is a draw from the product distribution $D_1 \times D_2 \times \dots \times D_m$.  Suppose, for this intuition-building only, that the mechanism accesses the data in $m$ epochs, first accessing block $B_1$ of attributes for all $n$ individuals, then accessing block $B_2$ of attributes for all $n$ individuals, and so on.  At epoch~$i\in [m]$ the analyst may carry out any $(\epsilon_0,\delta_0)$-DP analysis of the data on block~$i$.  In this case, we claim we can apply the transfer theorem $m$ times while retaining the accuracy guarantees and paying a factor of $m$ in the failure probability~$\beta$.
To see this, note that, because of the independence assumptions, we can assume that the data for block~$B_i$ have not even been selected before processing of this block.  In this case, an accuracy adversary -- even one with all the data of blocks $B_1,\dots,B_{i-1}$ "hard-wired" in, is just an arbitrary adversary. Allowing this adversary to interact with an independently randomly chosen block $B_i$ is precisely what happens in differential privacy: an adversary interacts with (apparently) freshly drawn data.  We can therefore apply the transfer theorem to conclude that, on this $i$th block, with probability at least $1-\beta$, the responses are $\alpha$-accurate.  A union bound then gives the result, yielding an upper bound of $m\beta$ on the probability of failure. 

While this "thick" streaming access mode is not required for our algorithms, it remains useful for building intuition when we depart from the full independence data models.

For our most general result, we consider models in which correlations between attributes $a_i$ and $a_j$ in the data of a single individual falls exponentially with their "distance" $|i-j|$, and we restrict the "width" of a query so that it cannot simultaneously access very distant elements. Roughly speaking, in our model distant attributes have high probability of being independent and vanishing probability of being arbitrarily dependent. We show that we can again re-use the privacy budget, paying only a small additional probability of failure due to the low-probability dependence events.


\begin{theorem}
(Informal) Suppose the probability that two attributes at distance $d$ are not independent is negligible, and suppose further that queries involve only attributes with distance at most $d$. Then, if our mechanism $M$ is $(\epsilon,\delta)$-DP and $(\alpha',\beta')$-sample accurate on every sequence of $2d+1$ consecutive attributes, it's also $(\alpha,m\beta + \text{negl})$-distributionally accurate where $(\alpha = \alpha(\epsilon,\delta,\alpha'),\beta = \beta(\epsilon,\delta,\beta'))$ are the parameters we get from the transfer theorem.
\end{theorem}



\section{Preliminaries}
\label{sec:preliminaries}


We are interested in query answering {\em mechanisms} that operate on datasets and produce outputs.  A standard view is that the mechanism interacts with an {\em adversary} whose goals are unknown and who may be malicious.  Both parties may employ randomness.

The interaction between a mechanism $\mathcal{M}$ and an adversary $A$ using sample $S$, is a random variable denoted by $\Interact(\mathcal{M},A;S)$, where the adversary generates queries $q_i$ and the mechanism $\M$ generates responses $a_i$, giving rise to {\em transcripts} of the form $(q_1,a_1, q_2,a_2, \dots q_k,a_k)$. Later queries may be chosen as functions of the transcript prefix.  We will sometimes use the shorthand $I(S)$ when $\mathcal{M}$ and $A$ are clear from context.
%
The set of transcripts that can be generated by the interaction between $\mathcal{M}$ and $A$ will be denoted $\Interact(\mathcal{M},A,*)$.

In this work, individuals are represented in the dataset as a sequence of $m$ \emph{attributes}, or \emph{covariates}. Doing so allows us to formalize the idea of \emph{distance} among attributes in a dataset as the difference in the indices of the attributes.

\begin{definition}
Datasets $X$ and $X'$ of the same cardinality are \emph{adjacent} if they differ on at most one element.
\end{definition}

\begin{definition}
A mechanism $\mathcal{M}$ is $(\epsilon,\delta)$-differentially private if for any pair of adjacent datasets $X,X'$, any adversary $A$, and any set of transcripts $E$, we have

$$\Pr[\Interact(\mathcal{M}, A, X) \in E] \le e^\epsilon \cdot \Pr[\Interact(\mathcal{M}, A, X') \in E] + \delta ,$$
where the probability space is over the randomness of $\mathcal{M}$ and $\mathcal{A}$.
\end{definition}

\begin{definition}
A mechanism $\mathcal{M}$ satisfies $(\alpha, \beta)$-sample accuracy if for every data analyst $A$ and every data distribution $\mathcal{P}$,

$$\Pr_{X \sim \mathcal{P}^n, \Interact(\mathcal{M}, A, X)} \left[\max_j |q_j(S) - a_j| \ge \alpha\right] \le \beta$$

Similarly, $\mathcal{M}$ satisfies $(\alpha,\beta)$-distributional accuracy if for every data analyst $A$ and every data distribution $\mathcal{P}$,

$$\Pr_{X \sim \mathcal{P}^n, \Interact(\mathcal{M}, A, X)} \left[\max_j |q_j(\mathcal{P}^n) - a_j| \ge \alpha\right] \le \beta$$

\end{definition}

\begin{definition}
We say that a sequence of random variables $(B_1, B_2, \dots, B_m)$ is $k$-dependent if for any two subsets $I$ and $J$ of $\{1,2, \dots, m\}$ such that $\max{(I)}<\min{(J)}$
and $\min(J) - \max(I) > k$, the families of random variables $(B_i)_{i \in I}$ and $(B_j)_{j \in J}$ are independent.
\end{definition}

\begin{definition}
A linear query (sometimes called statistical query) is a query $q$ such for any individual $X \in \mathcal{X}$, $q(x) \in [0,1]$, and for any sample $S \in \mathcal{X}^n$, $q(S) = \frac{1}{n} \sum_{x \in S} q(x)$
\end{definition}

From time to time, we will need to focus on the queries that involve a specific collection of attributes. For this purpose, we introduce the following definition:
\begin{definition}
\label{defn:restricted_mechanism}
Let $Q$ be a collection of queries, defined before the interaction happens. Given a mechanism $\mathcal{M}$, the transcript of the interaction restricted to $Q$ is defined as follows:

\begin{enumerate}
    \item $\mathcal{M}$ interacts with an adversary $A$, producing transcript $\Pi$
    \item As a postprocessing step, we remove every query and answer $(q,a)$ from $\Pi$ such that $q \notin Q$. Let $\Pi'$ denote resulting transcript.
    \item $\Pi'$ is the transcript of the interaction restricted to $Q$.
\end{enumerate}
\end{definition}

Intuitively, this is just "projecting" the transcript onto $Q$.

\subsection{Transfer Theorem}

The following is Theorem 3.5 from \cite{jung2019new}.

\begin{theorem}
Suppose $M$ is $(\epsilon,\delta)$-DP and $(\alpha,\beta)$-sample accurate for linear queries. Then for any data distribution $\mathcal{P}$, a sample $S \sim \mathcal{P}^n$, any analyst $\mathcal{A}$, and any constants $c,d > 0$:

$$\Pr_{S \sim \mathcal{P}^n, \Pi \sim \Interact(M,A;S)} \left[\max_j \left|a_j - q_j(\mathcal{P})\right| > \alpha + (e^\epsilon - 1) + c + 2d\right] \le \frac{\beta}{c} + \frac{\delta}{d}$$

i.e. it is $(\alpha',\beta')$-distributionally accurate for $\alpha' = \alpha + e^\epsilon - 1 + c + 2d$ and $\beta' = \frac{\beta}{c} + \frac{\delta}{d}$.
\end{theorem}

There are two facts to note here. Firstly, the transfer theorem assumes that all queries are linear queries (often called statistical queries in the literature). A linear query $q$ is one in which for each $x \in S$, $q(x) \in [0,1]$ and $q(S) = \frac{1}{n} \sum_{x \in S} q(x)$.

The notable features of a linear query are that $q$ must be a function of $x$, so it is deterministic and also cannot use information not captured in the features of the database, such as index. 
Linear queries are powerful; it is known that we can learn nearly everything that is PAC-learnable in the statistical queries learning model~\cite{kearns1998efficient}. In addition, there is a vast literature on handling very large numbers of differentially private statistical queries, beginning with the exciting contributions in~\cite{blum2013learning,hardt2010multiplicative}.

Note that, were we to remove the constraint that the query must be a function only of the covariates (and not, say the index of a row in the database), 
the sample accuracy of the mechanism would become ill-defined.



The other key fact is that, in the statement of the transfer theorem, the probability is taken over both the sample and the randomness employed during the interaction. Thus, the mechanism could be arbitrarily bad for some particularly unrepresentative sample. That is, we could come up with "counterexample" samples where we do get $a_j - q_j(\mathcal{P})$ to be very large (imagine a sample where $\alpha + \alpha'$ is significantly greater than $|q(S) - q(\mathcal{P})|$ for many queries $q$).

In the following sections, we will analyze mechanisms, where, to bound their privacy loss na\"{i}vely, we would need to take the composition of $m$ mechanisms, requiring us to pay an $\Omega(\sqrt{m})$ factor in the DP guarantee. By assuming (limited) independence in our data, we are able to instead bound the privacy loss with the composition of 1 or 2 mechanisms, while having the same $m$-fold increase in the probability of failure that we would get from composition.

\section{Full Independence}

In this setting, we are motivated by the structure of chromosomes. The entire sequence of DNA is contained in many linear chromosomes, and there is no known dependence between the sequence of one linear chromosome and the sequences of any other linear chromosomes. As such, it is reasonable to assume that these sequences are all independent. Thus, if we consider each linear chromosome to be a block, then we obtain the following bounds when doing adaptive data analysis with in a simple setting:

\begin{theorem}
\label{thm:independent}
Let $\mathcal{M}$ be a query answering mechanism $\mathcal{M}$, such that when given $(X_1, X_2, \dots X_n) \sim D^n$ for a population distribution $D$ such that the attributes are divided into fully independent blocks $B_1, B_2, \dots B_m$, given a data analyst $A$, $\mathcal{M}$ proceeds as follows:

\begin{itemize}
\item $\mathcal{M}$ refuses to answer queries that involve attributes in different blocks. 

\item $\mathcal{M}$ ensures that, for each block $B_i$, the interaction restricted to queries on the block $B_i$ is $(\epsilon, \delta)$-DP and $(\alpha,\beta)$ sample accurate. 
\end{itemize}

Then, for every $c,d > 0$, $\mathcal{M}$ is $(\alpha',\beta')$ distributionally accurate where $\alpha' = \alpha + e^\epsilon-1 + c+2d$ and $\beta' = m\left(\frac{\beta}{c} + \frac{\delta}{d}\right)$.
\end{theorem}

\begin{proof}
Let $X = (X_1, X_2, \dots, X_n)$ denote the sample that $\mathcal{M}$ takes as input. For each $i$, we conduct a thought experiment to define a query answering mechanism $\mathcal{M}'_i$ as follows:

$\mathcal{M}'_i$ takes as data the $i^{th}$ block of $X$ (which we denote $X^{(i)}$). Then, $\mathcal{M}'_i$ samples new values for blocks $B_1, B_2, \dots, B_{i-1}, B_{i+1}, \dots, B_m$ from $D$\footnote{The reason why this is just a thought experiment is that in reality the mechanism will not know the distribution $D$.  This is why we carry out data analysis in the first place.}. Let $X'$ denote this new sample. $\mathcal{M}'_i$ then interacts with an analyst $A$ by running $\mathcal{M}$ with the new sample $X'$. The queries on any block other than $B_i$ update the states of $A$ and $\mathcal{M}'_i$, but are not considered to be queries and answers of the interaction between $A$ and $\mathcal{M}'_i$.

Now, by definition, when $A$ interacts with $\mathcal{M}'_i$, only queries on the $i^{th}$ block interact with the data in any way, which means this interaction is $(\epsilon,\delta)$-DP. Furthermore, it is $(\alpha,\beta)$-sample accurate from the assumption that $\mathcal{M}$ was $(\alpha,\beta)$-sample accurate for the queries on block $B_i$. Thus, by theorem 3.5 from \cite{jung2019new}, $\mathcal{M}'_i$ is $(\alpha',\beta'/m)$ distributionally accurate.

Now, since $B_i$ is independent from all other blocks, $X' \sim D^n$. Thus, all $\mathcal{M}'_i$ does is interact with $A$ as if it were $\mathcal{M}$ on sample $X'$, except it only writes queries on block $B_i$ on the transcript. When we consider the distribution with randomness over the choice of sample, the mechanism, and the adversary, the distribution of transcripts produced by $\Interact(\mathcal{M}'_i, A, X^{(i)})$ is therefore exactly the same as the distribution of transcripts produced by $\Interact(\mathcal{M}, A, X)$, with the added postprocessing step of throwing away every query and answer asked about some block other than $B_i$.

Thus, the distribution of transcripts produced by $\Interact(\mathcal{M}, A, X)$ is identical to the distribution of the concatenation of the transcripts of $\Interact(\mathcal{M}'_i, A_i, X^{(i)})$ for every $i$ where all of the $A_i$ are copies of $A$. Taking a union bound over the accuracy guarantees for the latter, we get that $\mathcal{M}$ is $(\alpha',\beta')$ accurate.
\end{proof}

\section{Partial Independence}
This model is a generalization of the previous model, as the intuition that attributes which are close to one another can be related produces data which do not satisfy the assumptions necessary for the full independence model (consider items that are close, but on different sides of a block boundary). We therefore generalize our result to the case where adjacent blocks are allowed to be related. Additionally, we restrict access to the data to a streaming model. This allows us to achieve stronger accuracy guarantees; specifically, we obtain a bound with twice the privacy loss of full independence; without the streaming restriction it would be thrice the privacy loss.

To do this, we first introduce the following lemma that we will use in the proof. Intuitively, the lemma states that a transformation of individuals preserves privacy.


\begin{lemma}
\label{lemma: changing_db}
Let $\my$ be an $\ed$-differentially private mechanism with data domain $\Y$.  Then the mechanism $\mx$, defined next and having data domain $\X$, is also $\ed$-differentially private. 


$\mx$ takes as input a database $X\in \X^n $ and constructs $Y = f(X)\in \Y^n$, where $f$ is a randomized mapping $f:\X \rightarrow \Y$. The randomness is chosen independently every time $f$ is called, and we define $Y=f(X) = \{f(x) \mid x \in X\}$. Then, $\mx$ runs $\my$ on $Y$: given (oracle) access to any adversary $\A$, $\mx$ simply acts as a channel, conveying queries from $\A$ to $\my$ and responses from $\my$ to $\A$.
\end{lemma}

\begin{proof}
Fix an adversary $\A$, and let $\Pi$ be the random variable denoting the transcript of the interaction between $\A$ and $\my$; that is, $\Pi \in \Interact(\my,\A,*)$, the set of all transcripts that can be produced by these two parties.

Let $Q$ be a random variable that represents the value of the database given to $\my$ by $\mx$, with randomness over $X$ and $f$. Since $\my$ is $\ed$-differentially private we have that, for any event $E \in \Interact(\my,\A,*)$ and any $Y'$ adjacent to $Y$,
$$\Pr[\Pi \in E \mid Q = Y] \le e^\epsilon \Pr[\Pi \in E \mid Q = Y'] + \delta ,$$
where the probabilities are over the randomness of $\my$ and~$\A$.

Fix an adjacent pair $X$ and $X'$ in $\X^n$ and let $i$ be the index in which they differ.
For $R \in \{X,X'\}$ we have:

$$\Pr[\Pi \in E \mid R_i=X_i] = \sum_{y \in \Y} \Pr[\Pi \in E \mid Q_i = y] \cdot \Pr[Q_i = y \mid R_i= X_i]$$
since the event $\Pi \in E$ is independent of the original database $R$ conditional on the transformed database $Y$. 
Here the probabilities are over the randomness in the mapping $f$ and the randomness in the $[\my,\A]$ interaction, {\it i.e.}, the coin flips of $\my$ and $\A$.

Let $y^*$ denote the outcome which minimizes $\Pr[\Pi \in E \mid Q_i = y^*]$. Additionally, recall that we defined $Y$ as the input to $\my$, so if we fix $Y_i = y)$, then $Y = (f(X_{-i}), y)$.


\begin{align}
    &\Pr_{\Interact(\mathcal{M}^{\mathcal{X}},A, X)}[\Pi \in E \mid R = X]\\
    &= \sum_{y\in\Y} \Pr_{f(X_{-i}),\Interact(\mathcal{M}^{\mathcal{Y}}, A, Y)}[\Pi \in E \mid Q_i = y] \cdot \Pr_{f(X_i)}[Q_i = y \mid R_i=X_i]\\
    &\le \sum_{y} \left(e^\epsilon \Pr[\Pi \in E \mid Q_i = y^*] + \delta\right) \cdot \Pr[Q_i = y \mid R_i=X_i]\\
    &= \left(e^\epsilon\Pr[\Pi \in E \mid Q_i = y^*] + \delta\right) \sum_{y}  \Pr[Q_i = y \mid R_i=X_i]\\
    &= \left(e^\epsilon\Pr[\Pi \in E \mid Q_i = y^*] + \delta\right) \sum_{y}  \Pr[Q_i = y \mid
    R_i=X'_i]\\
    &\le \sum_{y} (e^\epsilon \Pr[\Pi \in E \mid Q_i = y] + \delta) \Pr[Q_i = y \mid R_i=X'_i]\\
    &=\delta + e^\epsilon \sum_{y} \Pr[\Pi \in E \mid Q_i = y] \Pr[Q_i = y \mid R_i=X'_i]\\
    &= e^\epsilon \Pr[\Pi \in E \mid R=X'] + \delta
\end{align}

Since $Y_{-i}$ is sampled independently from $Y_i$ and $X_i$, the inequality in line (3) holds when we condition on any value of $Y_{-i}$ by definition of $(\epsilon,\delta)$-DP, so it must also hold when we take the probability over $Y_{-i}$ as well. The equality in line $(5)$ follows by the law of total probability. 
\end{proof}

\begin{theorem}
\label{thm:blocks}
Suppose we have a query answering mechanism $\mathcal{M}$, such that when given $(X_1, X_2, \dots X_n) \sim D^n$ for a population distribution $D$ where the attributes are grouped into 1-dependent blocks $\{B_1, B_2, \dots B_m\}$(sequences of consecutive attributes), and a stateful data analyst $A$, ${\cal M}$ procedes as follows:

At each time step $t \in [m]$, $\mathcal{M}$ has an arbitrary $(\epsilon, \delta)$-DP interaction with $A$ in which $A$ asks linear queries about block $B_t$ and $\mathcal{M}$ answers the queries in such a way that the interaction is $(\alpha,\beta)$ sample accurate. The transcript is denoted by $S_t$.

Then, for every $c,d > 0$, $\mathcal{M}$ is $(\alpha',\beta')$ accurate where $\alpha' = \alpha + e^{2\epsilon}-1 + c+2d$ and $\beta' = m\left(\frac{\beta}{c} + \frac{2\delta}{d}\right)$.
\end{theorem}


\begin{proof}
First, for each $i\in [m]$, we define a query answering mechanism $\mathcal{M}'_i$ and adversary $A'_i$ as follows:




$\mathcal{M}'_i$ takes as input the $i^{th}$ block of our original sample of $n$ individuals $(X_1, X_2, \dots X_{n}) \sim D^{n}$, which we will denote $X^{(i)}$. It then resamples the first $i-1$ blocks from $D^n$ conditional on $X^{(i)}$. We will refer to this database of the $i-1$ resampled blocks and the $i^{th}$ block as $Y$. Then, $A'_i$ and $\mathcal{M}'_i$ run $\Interact(\mathcal{M},A,Y)$ for $t$ from $1$ to $i$, and we denote the transcript generated at time $t$ by this interaction as $S'_i$. While both parties may keep track of $S'_1, S'_2, \dots S'_{i-1}$, only $S_i = S'_i$ is considered to be the transcript of this interaction.

Now, we note that the distribution of transcripts $S_1, S_2, \dots, S_i$ produced by $\mathcal{M}'_i$ and $\mathcal{M}$ are identical. This is because, analogously to the proof of theorem~\ref{thm:independent}, first sampling a block and then sampling the rest of the data conditional on that block produces the same distribution as sampling all of the data at once.


Now, we shall analyze the accuracy of $\mathcal{M}'_i$. By definition, the first $i-2$ blocks are independent of $X^{(i)}$, so the part of $\Interact(\mathcal{M}'_i, A'_i, X^{(i)})$ that generates $S'_1, S'_2, \dots S'_{i-2}$ is independent of $X^{(i)}$ and thus does not incur any privacy loss with respect to $X^{(i)}$.

For $S_{i-1}$, recall that $(X_1, X_2, \dots, X_n)$ are drawn from the distribution iid. Thus, when we fix $X^{(i)}$ and resample the $i-1^{st}$ block conditional on $X^{(i)}$, the value of the $i-1^{st}$ block of each individual $X_j$ is a randomized mapping of the $i^{th}$ block the same individual $X_j$, independent of every other individual $X_{j'}$. Then, the interaction between $\mathcal{M}'_i$ and $A'_i$ on block $B_{i-1}$ is $(\epsilon,\delta)$-DP with respect to the resample $i-1^{st}$ block. Thus, by Lemma~\ref{lemma: changing_db}, the part of $\Interact(\mathcal{M}'_i, A'_i, X^{(i)})$ that generates $S_{i-1}$ is $(\epsilon,\delta)$-DP.

Finally, because $\mathcal{M}$ is $(\epsilon,\delta)$-DP on the interaction in each block, the part of $\Interact(\mathcal{M}'_i, A'_i, X^{(i)})$ that generates $S_i$ is $(\epsilon,\delta)$-DP. As such, $\mathcal{M}'_i$ is $(2\epsilon, 2\delta)$-DP and $(\alpha, \beta)$-sample accurate. By theorem 3.5 from \cite{jung2019new}, $\mathcal{M}'_i$ is $(\alpha',\beta'/m)$ distributionally accurate.

This tells us that $\mathcal{M}'_i$ is $(\alpha', \beta'/m)$ distributionally accurate for each $i$, and just like in Theorem~\ref{thm:independent}, we can concatenate the transcripts $S_i$ computed from $\mathcal{M}'_i$ for each i to get the transcript $S_1, S_2, \dots, S_m$ with the same distribution as the interaction between $\mathcal{M}$ and $A$. taking a union bound over the probabilities of failure over these $m$ mechanisms tells us that $\mathcal{M}$ is $(\alpha',\beta')$ distributionally accurate.

\end{proof}

\section{Exponential Decay}
Our final model directly captures the idea that the strength of the relationship between two attributes should be decreasing with the distance between them. We model this via following definition:

\begin{definition}
\label{defn: decaying_correlation}
In the decaying correlation model with parameter $p$, we are given attributes $B_1, B_2, \dots B_n$, such that for each $i$, $B_i$ and $B_{i+1}$ are independent with probability $p$, and otherwise they are arbitrarily related. The event of $B_i$ and $B_{i+1}$ being related and $B_j$ and $B_{j+1}$ being related are independent for all $i \ne j$, and for any $i < j$, $B_i$ and $B_j$ are related iff $B_{i'-1}$ is related to $B_{i'}$ for every $i < i' \le j$.
\end{definition}

With this model, there is some dependence between all of the attributes. However, due to the way it is defined, the dependence only exists with small probability over the sample between distant attributes. Thus, we can utilize similar arguments as above, and simply add this small probability to the probability of failure.
%
%
%
%



\begin{theorem}
\label{thm:expo}
(General Access)
Given a database $X$ in the decaying correlation model with parameter $p$ and $m$ attributes, a mechanism $\mathcal{M}$ which satisfies the following properties while interacting with an adversary $A$ is $(\alpha', \beta')$-distributionally accurate where for all integers $d > 0$:
$$\alpha' = \alpha + (e^{\epsilon}-1) + c + 2f, \quad \beta' = m\left(\frac{\beta}{c} + \frac{\delta}{f}\right) + 2n(1-p)^{d+1}$$

\begin{enumerate}
\item For each $i$, $\mathcal{M}$ restricted to queries that involve at least one of the attributes $\{B_{i-2d}, B_{i-2d+1}, \dots B_{i+2d}\}$ is $(\epsilon,\delta)$-DP.

\item For each $i$, $\mathcal{M}$ restricted to queries that involve only attributes in the set $\{B_{i-d}, B_{i-d+1}, \dots B_{i+d}\}$ is $(\alpha,\beta)$ sample accurate.

\item Any query can only involve attributes $B_i$ and $B_j$ if $|i-j| \le d$.

\end{enumerate}
\end{theorem}

\begin{proof}
Let $D$ be the population distribution. For each $i$, we define a query answering mechanism $\mathcal{M}'_i$ as follows:

$\mathcal{M}'_i$ takes as data the attributes $\{B_{i-d}, \dots, B_{i+d}\}$ of $n$ individuals $(X_1, X_2, \dots X_{n}) \sim D^{n}$, which we shall refer to as $X^{(i)}$. $\mathcal{M}'_i$ then constructs $Y$ by sampling the attributes\\
$\{B_{i-2d}, B_{i-2d+1}, \dots, B_{i-d-1}, B_{i+d+1}, \dots, B_{i+2d}\}$ for $n$ individuals from the population $D$ conditional on agreeing with $X^{(i)}$ on the attributes $\{B_{i-d}, \dots, B_{i+d}\}$. The rest of the attributes for these $n$ individuals are sampled from $D$ independently from $X^{(i)}$.

Then, $\mathcal{M}'_i$ interacts with an adversary $A$ by simulating $\mathcal{M}$ on the dataset $Y$. Any query which asks about an attribute outside of the set $\{B_{i-d}, \dots, B_{i+d}\}$ still takes place in the interaction, but it is not recorded in the transcript.

This construction guarantees that our $(\alpha,\beta)$-sample accuracy bound on $\mathcal{M}$ restricted to queries that involve at least one of the attributes $\{B_{i-d}, B_{i-d+1}, \dots, B_{i+d}\}$ also applies to $\mathcal{M}'_i$, since $\{B_{i-d}, \dots, B_{i+d}\}$ are exactly the attributes $\mathcal{M}'_i$ takes as data, so sample accuracy is well-defined over these queries.

The privacy loss of $\mathcal{M}'_i$ can be bounded by the privacy loss when we only consider queries that involve at least one of the attributes $\{B_{i-2d}, B_{i-2d+1}, \dots, B_{i+2d}\}$ since all of the other attributes are sampled independently from the data. We are given that this is $(\epsilon,\delta)-DP$.

Thus, $\mathcal{M}'_i$ is $(\epsilon, \delta)-DP$ and $(\alpha,\beta)$-sample accurate. By the transfer theorem, $\mathcal{M}'_i$ on the set of queries involving attribute $B_i$ is $(\alpha',\beta_2)$-distributionally accurate for

$$\alpha' = \alpha + (e^{\epsilon}-1) + c + 2f, \quad \beta_2 = \frac{\beta}{c} + \frac{\delta}{f}$$




Now, by construction, if we condition on $Y \sim X$, we can get the same distribution of transcripts as $\Interact(\mathcal{M}'_i, A, X^{(i)})$ by computing the transcript of $\Interact(\mathcal{M}, A, X)$ restricted to queries that involve only attributes in the set $\{B_{i-d}, B_{i-d+1}, \dots, B_{i+d}\}$. Additionally, by assumption 2, we know that the guarantee for $\mathcal{M}'_i$ applies to every query that involves attribute $B_i$. As such, $(\alpha', \beta_2)$ bounds the distributional accuracy of all queries involving attribute $B_i$ in $\Interact(\mathcal{M}, A, X)$. Thus, we can bound the distributional accuracy of $\mathcal{M}$ by union bounding the probability that the distributional error of any answer in any of $\{\mathcal{M}'_{1}, \mathcal{M}'_{2}, \dots \mathcal{M}'_m\}$ is greater than $\alpha'$, conditional on $Y \sim X$.

We get $Y \sim X$ iff $X$ satisfies the property that all attributes outside of $\{B_{i-2d}, B_{i-2d+1}, \dots, B_{i-2d}\}$ are independent from all attributes in the set $\{B_{i-d}, \dots, B_{i+d}\}$. This happens iff $B_{i-2d-1}$ is independent from $B_{i-d}$ and $B_{i+2d+1}$ is independent from $B_{i+d}$ for every individual in $X$. This probability is at least $1-2n(1-p)^{d+1}$ by taking a union bound over the 2 attributes $B_{i-2d-1}$ and $B_{i+2d+1}$ for each of the $n$ individuals.

As such we can bound the accuracy of the answers $\mathcal{M}$ produces to the queries involving some attribute in the set $\{B_{i-d}, B_{i-d+1}, \dots B_{i+d}\}$ by simply adding the probability that it does not produce the same distribution of transcripts as $\mathcal{M}'_i$ to the probability of failure, so it is $(\alpha',\beta')$-distributionally accurate for

$$\beta' = m\beta_2 + 2n(1-p)^{d+1}$$

or equivalently,

$$\alpha' = \alpha + (e^{\epsilon}-1) + c + 2f, \quad \beta' = m\left(\frac{\beta}{c} + \frac{\delta}{f}\right) + 2n(1-p)^{d+1}$$
as desired.
\end{proof}

We can improve the parameters by constraining access to the {\em sliding window} model studied in other contexts (see, for example, the tutorial~\cite{hirzel2017sliding} on sliding window aggregation algorithms, and the references therein).  Details may be found in the appendix.

\section{Using the Label in the Mechanism}

In this Section, we show that, at a small cost
in accuracy, we can extend our results to analyses that incorporate the labels.  This is a pleasant surprise, as the labels are "morally" exposed to high privacy loss. The key idea to note here is that even though we use the exact marginal distribution of the label, which cannot be done privately, the query-answering mechanisms that we use as sub-processes take data without the label, for which no information has been revealed to the adversary.

\begin{theorem}
Suppose the following is true:

\begin{enumerate}
    \item There is a binary attribute $y$ which we refer to as the "label."
    \item We have a mechanism $\mathcal{M}_0$ which is $(\alpha_0, \beta_0)$-distributionally accurate when $y=0$ for every individual in the distribution.
    \item We have a mechanism $\mathcal{M}_1$ which is $(\alpha_1, \beta_1)$-distributionally accurate when $y=1$ for every individual in the distribution.
\end{enumerate}

Now, consider the mechanism $\mathcal{M}$ which on input $S$, runs as follows:

\begin{enumerate}
    \item Partition $S$ into samples $S_0 = \{s \in S \mid s \text{ has } y=0\}$ and $S_1 = \{s \in S \mid s \text{ has } y=1\}$
    \item When $\mathcal{M}$ receives query $q$ from the adversary, it asks $q$ to $\mathcal{M}_0$ on sample $S_0$ and gets answer $a_0$. It then asks $q$ to $\mathcal{M}_1$ on sample $S_1$ and gets answer $a_1$. $\mathcal{M}$ then returns the answer
    $$a_0 \frac{|S_0|}{|S|} + a_1 \frac{|S_1|}{|S|}.$$
\end{enumerate}

Let $D$ be the population distribution, $D_y$ be the marginal distribution of the label $y$, and $p = \Pr_{y \sim D_y}[y=0]$. Then, $\mathcal{M}$ is $(\alpha, \beta)$-distributionally accurate for any $\delta > 0$ and

$$\alpha = p\alpha_0 + (1-p)\alpha_1 + \frac{\delta p}{\sqrt{n}}, \quad \beta = \beta_0 + \beta_1 + 2e^{-2\delta^2}$$
\end{theorem}

\begin{proof}

To approximate the population proportion, we want to take $p$ times the output of $\mathcal{M}_0$ plus $1-p$ times the output of $\mathcal{M}_1$. To see this, if we let $D_0$ be the population distribution when we let $y=0$, and $D_1$ be the population distribution when we let $y=1$, then we have for any query $q$, $p q(D_0) + (1-p) q(D_1) = q(D)$. Thus, for query $q_j$, if we let $a_j$ be the answer from $\mathcal{M}_0$ and $a'_j$ be the answer from $\mathcal{M}_1$, we have

\begin{align*}|p a_j + (1-p) a'_j - q_j(D)| &= \left|p\left(a_j - q_j(D_0) \right) + (1-p)\left(a'_j - q_j(D_1)\right)\right|\\
&\le p|a_j - q_j(D_0)| + (1-p)|a'_j - q_j(D_1)|\end{align*}

Now, if we let $\hat{p} = \frac{|S_0|}{|S|}$, then we have by the triangle inequality

\begin{align*}
    |\hat{p} a_j + (1-\hat{p}) a'_j - q_j(D)| &\le |\hat{p} a_j + (1-\hat{p}) a'_j - pa_j - (1-p) a'_j| + |p a_j + (1-p) a'_j - q_j(D)|\\
    &\le |(\hat{p} - p)(a_j - a'_j)| + p|a_j - q_j(D_0)| + (1-p)|a'_j - q_j(D_1)|\\
    &\le |(\hat{p} - p)| + p|a_j - q_j(D_0)| + (1-p)|a'_j - q_j(D_1)|
\end{align*}
where the last inequality comes from the fact that the answers are bounded betweeen $[0,1]$. Now, $\hat{p} \sim \frac{1}{n} \text{binom}(n,p)$, so we can apply Chernoff to get that for any $\delta > 0$,

$$\Pr\left[|p - \hat{p}| < \frac{\delta p}{\sqrt{n}}\right] < 2e^{-2\delta^2}$$

Furthermore, by assumption, we know that $|a_j - q_j(D_0)| \le \alpha_1$ with probability $1-\beta_1$, and $|a'_j - q_j(D_1)| \le \alpha_2$ with probability $1-\beta_2$. Thus, taking a union bound, we get that for any $\delta > 0$, $\mathcal{M}$ is $(\alpha, \beta)$-sample accurate for

$$\alpha = p\alpha_0 + (1-p)\alpha_1 + \frac{\delta p}{\sqrt{n}}, \quad \beta = \beta_0 + \beta_1 + 2e^{-2\delta^2}$$
\end{proof}

\section{Discussion}

It is common practice in other fields to consider restricted classes of adversaries, where it is often possible to obtain better bounds.  For example, while Byzantine Agreement requires $n \ge 3t+1$ processors if the number of arbitrary failures can be as large as $t$, it requires only $n \ge t+1$ processors to handle $t$ fail-stop faults.  Similarly, in cryptographic protocols the bounds for {\em honest-but-curious} adversaries are often better than for the case of processors that diverge arbitrarily from the protocol.

This history, combined with the fact that an algorithm that only protects benign data analysts could still be of use, naturally leads to the question of whether it is possible to get better accuracy/adaptivity tradeoffs for more benign adaptive accuracy adversaries. Efforts to define an appropriate class of benign failure modes were stymied, however, by Freedman's paradox, which states that when we have a dataset of $n$ individuals and $n$ attributes, all of which are independent of a label $y$, we will find some attribute which is strongly correlated with $y$ with high probability. We feel this gives an example of a very natural error, na\"{i}ve but not malicious~\cite{Freedman}.  

Our conclusion is that some restriction -- {\it e.g.}, on data models or access models -- is therefore required, which led to this work.  It would be interesting to find other natural restrictions that lead to improvements comparable to -- or better than -- those obtained in this work.

\bibliography{biblio}

\appendix
\section{Sliding Window Model for Exponential Decay}

\begin{remark}
The form of this bound looks mostly identical to the bound in Theorem~\ref{thm:expo}, with a slightly better probability of failure. However, one must note that the privacy guarantee is now restricted to the set $\{B_{i-2d}, B_{i-2d+1}, \dots, B_{i+d}\}$ rather than $\{B_{i-2d}, B_{i-2d+1}, \dots, B_{i+2d}\}$ as it was before, so this does in fact give us a multiplicative constant improvement over Theorem~\ref{thm:expo}.
\end{remark}

\begin{theorem}
(Sliding Window)
Given a database $X$ in the decaying correlation model with parameter $p$ and $m$ attributes, a mechanism $\mathcal{M}$ which satisfies the following properties while interacting with an adversary $A$ is $(\alpha', \beta')$-distributionally accurate where
$$\alpha' = \alpha + (e^{\epsilon}-1) + c + 2f, \quad \beta' = m\left(\frac{\beta}{c} + \frac{\delta}{f}\right) + n(1-p)^{d+1}$$

\begin{enumerate}
\item For each $i$, $\mathcal{M}$ restricted to queries that involve only attributes in the set $\{B_{i-2d}, B_{i-2d+1}, \dots B_{i+d}\}$ is $(\epsilon,\delta)$-DP.

\item For each $i$, $\mathcal{M}$ restricted to queries that involve $B_i$ is $(\alpha,\beta)$-sample accurate.

\item Any query can only involve attributes $B_i$ and $B_j$ if $|i-j| \le d$.

\item After answering a query involving attribute $B_i$, the mechanism can no longer answer queries involving attributes $B_1, B_2, \dots B_{i-d}$.

\end{enumerate}
\end{theorem}

\begin{proof}
We define $X^{(i)}$ and $\mathcal{M}'_i$ as in theorem~\ref{thm:expo}, except we now stop the interaction immediately after $A$ asks the first query which involves an attribute in the set $\{B_{i+d+1}, \dots, B_{m}\}$ and before $\mathcal{M}'_i$ answers.

This interaction still contains every query which involves attribute $B_i$ by assumption 4, and these queries are all well-defined by assumption 3, so analogously to in theorem~\ref{thm:expo}, $\mathcal{M}'_i$ is $(\alpha,\beta)$-sample accurate.

This time, the privacy loss of $\mathcal{M}'_i$ can be bounded by the privacy loss when we only consider queries that involve the attributes $\{B_{i-2d}, B_{i-2d+1}, \dots, B_{i+d}\}$ since there are no queries asked about $\{B_{i+d}, B_{i+d+1}, \dots, B_{i+2d}\}$. We are given that this is $(\epsilon,\delta)-DP$.

Thus, $\mathcal{M}'_i$ is $(\epsilon, \delta)-DP$ and $(\alpha, \beta)$-sample accurate on all the queries in the transcript. Hence, by the transfer theorem, $\mathcal{M}'_i$ on the set of queries involving attribute $B_i$ is $(\alpha',\beta_2)$-distributionally accurate for

$$\alpha' = \alpha + (e^{\epsilon}-1) + c + 2f, \quad \beta_2 = \frac{\beta}{c} + \frac{\delta}{f}$$

In this setting, we cannot have any query involving $\{B_{i+d+1}, \dots, B_m\}$ be answered by $\mathcal{M}'_i$ or by $\mathcal{M}$ prior to any query involving $B_i$. Hence, this time, we note that the probability that some attribute in $\{B_1, B_2, \dots, B_{i-2d-1}\}$ is related to $B_{i-d}$ is at most $n(1-p)^{d+1}$ by taking a union bound over the $n$ individuals, in which case $\Interact(\mathcal{M}'_i,A,X^{(i)})$ restricted to queries that involve attribute $B_i$ produces the same distribution of transcripts as $\Interact(\mathcal{M}, A, X)$ restricted to queries that involve attribute $B_i$.

As such, similarly to in Theorem~\ref{thm:expo}, we can bound the accuracy of the answers in $\Interact(\mathcal{M},A,X)$ by adding the probability that $X$ has some attribute in $\{B_1, B_2, \dots, B_{i-2d-1}\}$ related to $B_{i-d}$ to the probability that any $\Interact(\mathcal{M}'_i,A,i)$ has an answer with error greater than $\alpha$. Thus, it is $(\alpha',\beta')$-distributionally accurate for

$$\alpha' = \alpha + (e^{\epsilon}-1) + c + 2f, \quad \beta' = m\left(\frac{\beta}{c} + \frac{\delta}{f}\right) + n(1-p)^{d+1}$$
\end{proof}

\end{document}